\documentclass[11pt]{article}

\usepackage{graphicx,geometry,verbatim,color}
\usepackage{amsmath}
\usepackage{amssymb}
\usepackage{float}

\floatstyle{ruled}
\newfloat{algorithm}{thp}{lop}[section]
\newenvironment{proof}{\noindent\textbf{Proof}}{\hfill\qed}
\newcommand{\qed}{\hfill$\Box$}
\newtheorem{lem}{Lemma}
\newtheorem{theo}{Theorem}
\newtheorem{defi}{Definition}
\newenvironment{theorem}[1]{\vspace{-0.25cm}\begin{theo}#1}{\end{theo}\vspace{-0.3cm}}
\newenvironment{definition}[1]{\vspace{-0.25cm}\begin{defi}#1}{\end{defi}\vspace{-0.3cm}}
\newenvironment{lemma}[1]{\vspace{-0.25cm}\begin{lem}#1}{\end{lem}\vspace{-0.3cm}}

\begin{document}

\title{The Impact of Topology on Byzantine Containment in Stabilization}

\author{Swan Dubois\protect\footnote{Universit\'e Pierre et Marie Curie \& INRIA, France, swan.dubois@lip6.fr} \footnote{Contact author, Telephone: 33 1 44 27 87 67, Postal address: LIP6, Case 26/00-225, Campus Jussieu, 4 place Jussieu, 75252 Paris Cedex 5, France} \and Toshimitsu Masuzawa\protect\footnote{Osaka University, Japan, masuzawa@ist.osaka-u.ac.jp} \and S\'{e}bastien Tixeuil\protect\footnote{Universit\'e Pierre et Marie Curie \& INRIA, France, sebastien.tixeuil@lip6.fr}}
\date{}

\maketitle

\begin{abstract}
Self-stabilization is an versatile approach to fault-tolerance since it permits a distributed system to recover from any transient fault that arbitrarily corrupts the contents of all memories in the system. Byzantine tolerance is an attractive feature of distributed system that permits to cope with arbitrary malicious behaviors. 

We consider the well known problem of constructing a maximum metric tree in this context. Combining these two properties prove difficult: we demonstrate that it is impossible to contain the impact of Byzantine nodes in a self-stabilizing context for maximum metric tree construction (strict stabilization). We propose a weaker containment scheme called \emph{topology-aware strict stabilization}, and present a protocol for computing maximum metric trees that is optimal for this scheme with respect to impossibility result.
\end{abstract}


\section{Introduction}

The advent of ubiquitous large-scale distributed systems advocates that tolerance to various kinds of faults and hazards must be included from the very early design of such systems. \emph{Self-stabilization}~\cite{D74j,D00b,T09bc} is a versatile technique that permits forward recovery from any kind of \emph{transient} faults, while \emph{Byzantine Fault-tolerance}~\cite{LSP82j} is traditionally used to mask the effect of a limited number of \emph{malicious} faults. Making distributed systems tolerant to both transient and malicious faults is appealing yet proved difficult~\cite{DW04j,DD05c,NA02c} as impossibility results are expected in many cases.

Two main paths have been followed to study the impact of Byzantine faults in the context of self-stabilization:
\begin{itemize}
\item \emph{Byzantine fault masking.} In completely connected synchronous systems, one of the most studied problems in the context of self-stabilization with Byzantine faults is that of \emph{clock synchronization}. In~\cite{BDH08c,DW04j}, probabilistic self-stabilizing protocols were proposed for up to one third of Byzantine processes, while in \cite{DH07cb,HDD06c} deterministic solutions tolerate up to one fourth and one third of Byzantine processes, respectively.
\item \emph{Byzantine containment.} For \emph{local} tasks (\emph{i.e.} tasks whose correctness can be checked locally, such as vertex coloring, link coloring, or dining philosophers), the notion of \emph{strict stabilization} was proposed~\cite{NA02c,SOM05c,MT07j}. Strict stabilization guarantees that there exists a \emph{containment radius} outside which the effect of permanent faults is masked, provided that the problem specification makes it possible to break the causality chain that is caused by the faults. As many problems are not local, it turns out that it is impossible to provide strict stabilization for those.
\end{itemize}

\noindent\textbf{Our Contribution.} In this paper, we investigate the possibility of Byzantine containment in a self-stabilizing setting for tasks that are global (\emph{i.e.} for with there exists a causality chain of size $r$, where $r$ depends on $n$ the size of the network), and focus on a global problem, namely maximum metric tree construction (see \cite{GS99c,GS03j}). As strict stabilization is impossible with such global tasks, we weaken the containment constraint by relaxing the notion of containment radius to containment area, that is Byzantine processes may disturb infinitely often a set of processes which depends on the topology of the system and on the location of Byzantine processes.

The main contribution of this paper is to present new possibility results for containing the influence of unbounded Byzantine behaviors. In more details, we define the notion of \emph{topology-aware strict stabilization} as the novel form of the containment and introduce \emph{containment area} to quantify the quality of the containment.
The notion of topology-aware strict stabilization is weaker than the strict stabilization but is stronger than the classical notion of self-stabilization (\emph{i.e.} every topology-aware strictly stabilizing protocol is self-stabilizing, but not necessarily strictly stabilizing).

To demonstrate the possibility and effectiveness of our notion of topology-aware strict stabilization, we consider \emph{maximum metric tree construction}. It is shown in \cite{NA02c} that there exists no strictly stabilizing protocol with a constant containment radius for this problem. In this paper, we provide a topology-aware strictly stabilizing protocol for maximum metric tree construction and we prove that the containment area of this protocol is optimal.

\section{Distributed System}

A \emph{distributed system} $S=(P,L)$ consists of a set
$P=\{v_1,v_2,\ldots,v_n\}$ of processes and a set $L$ of
bidirectional communication links (simply called links).
A link is an unordered pair of distinct processes.
A distributed system $S$ can be regarded as a graph whose vertex set is $P$
and whose link set is $L$, so we use graph terminology to describe a
distributed system $S$.

Processes $u$ and $v$ are called \emph{neighbors} if $(u,v)\in L$.
The set of neighbors of a process $v$ is denoted by $N_v$, and its
cardinality (the \emph{degree} of $v$) is denoted by $\Delta_v (=|N_v|)$.
The degree $\Delta$ of a distributed system $S=(P,L)$ is defined as
$\Delta = \max \{\Delta_v\ |\ v \in P\}$.
We do not assume existence of a unique identifier for each process.
Instead we assume each process can distinguish its neighbors from each other
by locally arranging them in some arbitrary order:
the $k$-th neighbor of a process $v$ is denoted by
$N_v(k)\ (1 \le k \le \Delta_v)$. The \emph{distance} between two processes
$u$ and $v$ is the length of the shortest path between $u$ and $v$.

In this paper, we consider distributed systems of arbitrary topology.
We assume that a single process is distinguished as a \emph{root},
and all the other processes are identical.

We adopt the \emph{shared state model} as a communication model
in this paper, where each process can directly read the states
of its neighbors.

The variables that are maintained by processes denote process states.
A process may take actions during the execution of the system. An
action is simply a function that is executed in an atomic manner
by the process.
The actions executed by each process is described by a finite set
of guarded actions of the form
$\langle$guard$\rangle\longrightarrow\langle$statement$\rangle$.
Each guard of process $u$ is a boolean expression involving
the variables of $u$ and its neighbors.

A global state of a distributed system is called a \emph{configuration}
and is specified by a product of states of all processes.
We define $C$ to be the set of all possible configurations
of a distributed system $S$.
For a process set $R \subseteq P$ and two configurations $\rho$ and $\rho'$,
we denote $\rho \stackrel{R}{\mapsto} \rho'$
when $\rho$ changes to $\rho'$ by executing an action of each process
in $R$ simultaneously.
Notice that $\rho$ and $\rho'$ can be different only in
the states of processes in $R$.
For completeness of execution semantics, we should clarify
the configuration resulting from simultaneous actions of
neighboring processes.
The action of a process depends only on its state
at $\rho$ and the states of its neighbors at $\rho$,
and the result of the action reflects on the state of the process
at $\rho '$.

A \emph{schedule} of a distributed system is an infinite sequence of
process sets.  Let $Q=R^1, R^2, \ldots$  be a schedule,
where $R^i \subseteq P$ holds for each $i\ (i \ge 1)$.
An infinite sequence of configurations
$e=\rho_0,\rho_1,\ldots$ is called an \emph{execution} from
an initial configuration $\rho_0$ by a schedule $Q$,
if $e$ satisfies $\rho_{i-1} \stackrel{R^i}{\mapsto} \rho_i$
for each $i\ (i \ge 1)$.
Process actions are executed atomically, and we also assume
that a \emph{distributed daemon} schedules the actions of processes,
i.e. any subset of processes can simultaneously execute
their actions.

The set of all possible executions from
$\rho_0\in C$ is denoted by $E_{\rho_0}$.
The set of all possible executions is denoted by $E$, that is,
$E=\bigcup_{\rho\in C}E_{\rho}$.
We consider \emph{asynchronous} distributed systems
where we can make no assumption
on schedules except that any schedule is \emph{weakly fair}:
every process is contained in infinite number of subsets
appearing in any schedule.

In this paper, we consider (permanent) \emph{Byzantine faults}:
a Byzantine process (i.e. a Byzantine-faulty process)
can make arbitrary behavior independently from its actions.
If $v$ is a Byzantine process,
$v$ can repeatedly change its variables arbitrarily.

\section{Self-Stabilizing Protocol Resilient to Byzantine Faults}\label{sec:stab}

Problems considered in this paper are so-called \emph{static problems}, 
i.e. they require the system to find static solutions.
For example, the spanning-tree construction problem is a static problem,
while the mutual exclusion problem is not.
Some static problems can be defined by a \emph{specification predicate}
(shortly, specification), $spec(v)$, for each process $v$:
a configuration is a desired one (with a solution) if 
every process satisfies $spec(v)$.
A specification $spec(v)$ is a boolean expression
on variables of $P_v~(\subseteq P)$ where $P_v$ is the set of processes
whose variables appear in $spec(v)$.
The variables appearing in the specification are
called \emph{output variables} (shortly, \emph{O-variables}).
In what follows, we consider a static problem defined by
specification $spec(v)$.

\noindent\textbf{Self-Stabilization.} A \emph{self-stabilizing protocol} (\cite{D74j}) is a protocol
that eventually reaches a \emph{legitimate configuration},
where $spec(v)$ holds at every process $v$, regardless of the initial configuration.
Once it reaches a legitimate configuration, every process never
changes its O-variables and always satisfies $spec(v)$.
 From this definition, a self-stabilizing protocol is expected to tolerate 
any number and any type of transient faults since it can eventually 
recover from any configuration affected by the transient faults.
However, the recovery from any configuration is guaranteed
only when every process correctly executes its action from 
the configuration, i.e., we do not consider existence of
permanently faulty processes.

\noindent\textbf{Strict stabilization.} When (permanent) Byzantine 
processes exist, Byzantine processes may not satisfy $spec(v)$.
In addition, correct processes near the Byzantine processes
can be influenced and may be unable to satisfy $spec(v)$.
Nesterenko and Arora~\cite{NA02c} define
a \emph{strictly stabilizing protocol} as a self-stabilizing protocol 
resilient to unbounded number of Byzantine processes.

Given an integer $c$, a \emph{$c$-correct process} is a process 
 defined as follows.

\begin{definition}[$c$-correct process]
A process is $c$-correct if it is correct (\emph{i.e.} not Byzantine) and located at distance more than $c$ from any Byzantine process.
\end{definition}

\begin{definition}[$(c,f)$-containment]
\label{def:cfcontained}
A configuration $\rho$ is \emph{$(c,f)$-contained} for specification
$spec$ if, given at most $f$ Byzantine processes, in any execution
starting from $\rho$, every $c$-correct process $v$ always satisfies $spec(v)$ and never changes
its O-variables.
\end{definition}

The parameter $c$ of Definition~\ref{def:cfcontained} refers to 
the \emph{containment radius} defined in \cite{NA02c}. 
The parameter $f$ refers explicitly to the number of Byzantine processes, 
while \cite{NA02c} dealt with unbounded number of Byzantine faults 
(that is $f\in\{0\ldots n\}$).

\begin{definition}[$(c,f)$-strict stabilization]
\label{def:cfstabilizing}
A protocol is \emph{$(c,f)$-strictly stabilizing} for specification
$spec$ if, given at most $f$ Byzantine processes, any execution
$e=\rho_0,\rho_1,\ldots$ contains a configuration $\rho_i$ that
is $(c,f)$-contained for $spec$.
\end{definition}

An important limitation of the model of \cite{NA02c}
is the notion of $r$-\emph{restrictive} specifications. 
Intuitively, a specification is $r$-restrictive if it prevents 
combinations of states that belong to two processes $u$ and $v$ 
that are at least $r$ hops away. 
An important consequence related to Byzantine tolerance is that 
the containment radius of protocols solving those specifications is 
at least $r$. 
For some problems, such as the spanning tree construction we consider
in this paper, $r$ can not be bounded to a constant.
We can show that there exists no $(o(n),1)$-strictly stabilizing
protocol for the spanning tree construction.

\noindent\textbf{Topology-aware strict stabilization.} In the former paragraph, we saw that there exist a number of impossibility results on strict stabilization due to the notion of $r$-restrictive specifications. To circumvent this impossibility result, we define here a new notion, which is weaker than the strict stabilization: the \emph{topology-aware strict stabilization} (denoted by TA-strict stabilization for short). Here, the requirement to the containment radius is relaxed, \emph{i.e.} the set of processes which may be disturbed by Byzantines ones is not reduced to the union of $c$-neighborhood of Byzantines processes but can be defined depending on the topology of the system and on Byzantine processes location.

In the following, we give formal definition of this new kind of Byzantine containment. From now, $B$ denotes the set of Byzantine processes and $S_B$ (which is function of $B$) denotes a subset of $V$ (intuitively, this set gathers all processes which may be disturbed by Byzantine processes).

\begin{definition}[$S_{B}$-correct node]
A node is \emph{$S_{B}$-correct} if it is a correct node (\emph{i.e.} not Byzantine) which not belongs to $S_{B}$.
\end{definition}

\begin{definition}[$S_{B}$-legitimate configuration]
A configuration $\rho$ is \emph{$S_{B}$-legitimate} for $spec$ if every $S_{B}$-correct node $v$ is legitimate for $spec$ (\emph{i.e.} if $spec(v)$ holds).
\end{definition}

\begin{definition}[$(S_{B},f)$-topology-aware containment]
\label{def:SfTAcontained}
A configuration $\rho_{0}$ is \emph{$(S_{B},f)$-topology-aware contained} for specification $spec$ if, given at most $f$ Byzantine processes, in any execution $e=\rho_0,\rho_1,\ldots$, every configuration is $S_{B}$-legitimate and every $S_B$-correct process never changes its O-variables. 
\end{definition}

The parameter $S_{B}$ of Definition~\ref{def:SfTAcontained} refers to the \emph{containment area}. Any process which belongs to this set may be infinitely disturbed by Byzantine processes. The parameter $f$ refers explicitly to the number of Byzantine processes.

\begin{definition}[$(S_{B},f)$-topology-aware strict stabilization]
\label{def:SfTAStrictstabilizing}
A protocol is \emph{$(S_{B},f)$-topology-\\\noindent aware strictly stabilizing} for specification $spec$ if, given at most $f$ Byzantine processes, any execution $e=\rho_0,\rho_1,\ldots$ contains a configuration $\rho_i$ that is $(S_{B},f)$-topology-aware contained for $spec$.
\end{definition}

Note that, if $B$ denotes the set of Byzantine processes and $S_{B}=\{v\in V|min\{d(v,b),b\in B\}\leq c\}$, then a $(S_{B},f)$-topology-aware strictly stabilizing protocol is a $(c,f)$-strictly stabilizing protocol. Then, a TA-strictly stabilizing protocol is generally weaker than a strictly stabilizing one, but stronger than a classical self-stabilizing protocol (that may never meet its specification in the presence of Byzantine processes).

The parameter $S_{B}$ is introduced to quantify the strength of fault containment, we do not require each process to know the actual definition of the set. Actually, the protocol proposed in this paper assumes no knowledge on this parameter.

\section{Maximum Metric Tree Construction}

In this work, we deal with maximum (routing) metric trees as defined in \cite{GS03j} (note that \cite{GS99c} provides a self-stabilizing solution to this problem). Informally, the goal of a routing protocol is to construct a tree that simultaneously maximizes the metric values of all of the nodes with respect to some total ordering $\prec$. In the following, we recall all definitions and notations introduced in \cite{GS03j}. 

\begin{definition}[Routing metric]
A \emph{routing metric} (or just \emph{metric}) is a five-tuple $(M,W,met,mr,$ $\prec)$ where:
\begin{enumerate}
\item $M$ is a set of metric values,
\item $W$ is a set of edge weights,
\item $met$ is a metric function whose domain is $M\times W$ and whose range is $M$,
\item $mr$ is the maximum metric value in $M$ with respect to $\prec$ and is assigned to the root of the system,
\item $\prec$ is a less-than total order relation over $M$ that satisfies the following three conditions for arbitrary metric values $m$, $m'$, and $m''$ in $M$:
\begin{enumerate}
\item irreflexivity: $m\not\prec m$,
\item transitivity : if $m\prec m'$ and $m'\prec m''$ then $m\prec m''$,
\item totality: $m\prec m'$ or $m'\prec m$ or $m=m'$.
\end{enumerate}
\end{enumerate}
Any metric value $m\in M\setminus\{mr\}$ satisfies the \emph{utility condition} (that is, there exists $w_0,\ldots,w_{k-1}$ in $W$ and $m_0=mr,m_1,\ldots,m_{k-1},m_{k}=m$ in $M$ such that $\forall i\in\{1,\ldots,k\},m_i=met(m_{i-1},w_{i-1})$).
\end{definition}

For instance, we provide the definition of three classical metrics with this model: the shortest path metric ($\mathcal{SP}$), the flow metric ($\mathcal{F}$), and the reliability metric ($\mathcal{R}$).

\[\begin{array}{rclrcl}
\mathcal{SP}&=&(M_1,W_1,met_1,mr_1,\prec_1)&\mathcal{F}&=&(M_2,W_2,met_2,mr_2,\prec_2)\\
\text{where}& & M_1=\mathbb{N}&\text{where}& & mr_2\in\mathbb{N}\\
&& W_1=\mathbb{N}&&& M_2=\{0,\ldots,mr_2\}\\
&& met_1(m,w)=m+w&&& W_2=\{0,\ldots,mr_2\}\\
&& mr_1=0&&& met_2(m,w)=min\{m,w\}\\
&& \prec_1 \text{ is the classical }>\text{ relation}&&& \prec_2 \text{ is the classical }<\text{ relation}
\end{array}\]
\[\begin{array}{rcl}
\mathcal{R}&=&(M_3,W_3,met_3,mr_3,\prec_3)\\
\text{where}& & M_3=[0,1]\\
&& W_3=[0,1]\\
&& met_3(m,w)=m*w\\
&& mr_3=1\\
&& \prec_3 \text{ is the classical }<\text{ relation}
\end{array}\]

\begin{definition}[Assigned metric]
An \emph{assigned metric} over a system $S$ is a six-tuple $(M,W,met,$ $mr,\prec,wf)$ where $(M,W,met,mr,\prec)$ is a metric and $wf$ is a function that assigns to each edge of $S$ a weight in $W$.
\end{definition}

Let a rooted path (from $v$) be a simple path from a process $v$ to the root $r$. The next set of definitions are with respect to an assigned metric $(M,W,met,mr,\prec,wf)$ over a given system $S$.

\begin{definition}[Metric of a rooted path]
The \emph{metric of a rooted path} in $S$ is the prefix sum of $met$ over the edge weights in the path and $mr$.
\end{definition}

For example, if a rooted path $p$ in $S$ is $v_k,\ldots,v_0$ with $v_0=r$, then the metric of $p$ is $m_k=met(m_{k-1},wf(\{v_k,v_{k-1}\})$ with $\forall i\in\{1,k-1\},m_i=met(m_{i-1},wf(\{v_i,v_{i-1}\})$ and $m_0=mr$.

\begin{definition}[Maximum metric path]
A rooted path $p$ from $v$ in $S$ is called a \emph{maximum metric path} with respect to an assigned metric if and only if for every other rooted path $q$ from $v$ in $S$, the metric of $p$ is greater than or equal to the metric of $q$ with respect to the total order $\prec$. 
\end{definition}
 
\begin{definition}[Maximum metric of a node]
The \emph{maximum metric of a node} $v\neq r$ (or simply \emph{metric value} of $v$) in $S$ is defined by the metric of a maximum metric path from $v$. The maximum metric of $r$ is $mr$. 
\end{definition}

\begin{definition}[Maximum metric tree]
A spanning tree $T$ of $S$ is a \emph{maximum metric tree} with respect to an assigned metric over $S$ if and only if every rooted path in $T$ is a maximum metric path in $N$ with respect to the assigned metric.
\end{definition}

The goal of the work of \cite{GS03j} is the study of metrics that always allow the construction of a maximum metric tree. More formally, the definition follow.

\begin{definition}[Maximizable metric]
A metric is \emph{maximizable} if and only if for any assignment of this metric over any system $S$, there is a maximum metric tree for $S$ with respect to the assigned metric.
\end{definition}

Note that \cite{GS99c} provides a self-stabilizing protocol to construct a maximum metric tree with respect to any maximizable metric. Moreover, \cite{GS03j} provides a fully characterization of maximizable metrics as follow.

\begin{definition}[Boundedness]
A metric $(M,W,met,mr,\prec)$ is \emph{bounded} if and only if: $\forall m \in M,\forall w\in W, met(m,w)\prec m \text{ or }met(m,w)=m$
\end{definition}

\begin{definition}[Monotonicity]
A metric $(M,W,met,mr,\prec)$ is \emph{monotonic} if and only if: $\forall (m,$ $m')\in M^2,\forall w\in W, m\prec m'\Rightarrow (met(m,w)\prec met(m',w)\text{ or }met(m,w)=met(m',w))$
\end{definition}

\begin{theorem}[Characterization of maximizable metrics \cite{GS03j}]
A metric is maximizable if and only if this metric is bounded and monotonic.
\end{theorem}

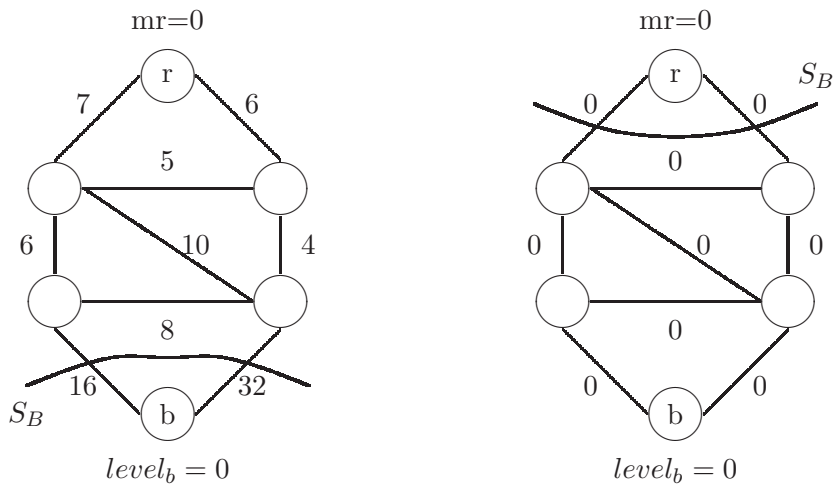
\begin{figure}[t]
\noindent \begin{centering} \ifx\JPicScale\undefined\def\JPicScale{0.75}\fi
\unitlength \JPicScale mm
\begin{picture}(165,95)(0,0)
\linethickness{0.3mm}
\put(50,85){\circle{10}}

\linethickness{0.3mm}
\put(30,65){\circle{10}}

\linethickness{0.3mm}
\put(70,65){\circle{10}}

\linethickness{0.3mm}
\put(30,45){\circle{10}}

\linethickness{0.3mm}
\put(70,45){\circle{10}}

\linethickness{0.3mm}
\put(50,25){\circle{10}}

\linethickness{0.3mm}
\multiput(30,70)(0.12,0.12){125}{\line(1,0){0.12}}
\linethickness{0.3mm}
\multiput(55,85)(0.12,-0.12){125}{\line(1,0){0.12}}
\linethickness{0.3mm}
\put(30,50){\line(0,1){10}}
\linethickness{0.3mm}
\multiput(30,40)(0.12,-0.12){125}{\line(1,0){0.12}}
\linethickness{0.3mm}
\multiput(55,25)(0.12,0.12){125}{\line(1,0){0.12}}
\linethickness{0.3mm}
\put(70,50){\line(0,1){10}}
\linethickness{0.3mm}
\put(35,65){\line(1,0){30}}
\linethickness{0.3mm}
\put(35,45){\line(1,0){30}}
\linethickness{0.3mm}
\multiput(35,65)(0.18,-0.12){167}{\line(1,0){0.18}}
\linethickness{0.3mm}
\put(140,85){\circle{10}}

\linethickness{0.3mm}
\put(120,65){\circle{10}}

\linethickness{0.3mm}
\put(160,65){\circle{10}}

\linethickness{0.3mm}
\put(120,45){\circle{10}}

\linethickness{0.3mm}
\put(160,45){\circle{10}}

\linethickness{0.3mm}
\put(140,25){\circle{10}}

\linethickness{0.3mm}
\multiput(120,70)(0.12,0.12){125}{\line(1,0){0.12}}
\linethickness{0.3mm}
\multiput(145,85)(0.12,-0.12){125}{\line(1,0){0.12}}
\linethickness{0.3mm}
\put(120,50){\line(0,1){10}}
\linethickness{0.3mm}
\multiput(120,40)(0.12,-0.12){125}{\line(1,0){0.12}}
\linethickness{0.3mm}
\multiput(145,25)(0.12,0.12){125}{\line(1,0){0.12}}
\linethickness{0.3mm}
\put(160,50){\line(0,1){10}}
\linethickness{0.3mm}
\put(125,65){\line(1,0){30}}
\linethickness{0.3mm}
\put(125,45){\line(1,0){30}}
\linethickness{0.3mm}
\multiput(125,65)(0.18,-0.12){167}{\line(1,0){0.18}}
\put(50,85){\makebox(0,0)[cc]{r}}

\put(140,85){\makebox(0,0)[cc]{r}}

\put(140,25){\makebox(0,0)[cc]{b}}

\put(50,25){\makebox(0,0)[cc]{b}}

\put(165,85){\makebox(0,0)[cc]{$S_B$}}

\put(50,95){\makebox(0,0)[cc]{mr=0}}

\put(140,95){\makebox(0,0)[cc]{mr=0}}

\put(50,15){\makebox(0,0)[cc]{$level_b=0$}}

\put(140,15){\makebox(0,0)[cc]{$level_b=0$}}

\put(35,80){\makebox(0,0)[cc]{7}}

\put(65,80){\makebox(0,0)[cc]{6}}

\put(50,70){\makebox(0,0)[cc]{5}}

\put(75,55){\makebox(0,0)[cc]{4}}

\put(55,55){\makebox(0,0)[cc]{10}}

\put(50,40){\makebox(0,0)[cc]{8}}

\put(25,55){\makebox(0,0)[cc]{6}}

\put(65,30){\makebox(0,0)[cc]{32}}

\put(35,30){\makebox(0,0)[cc]{16}}

\linethickness{0.3mm}
\qbezier(25,30)(25.43,30.23)(30.48,32.17)
\qbezier(30.48,32.17)(35.53,34.11)(40,35)
\qbezier(40,35)(42.53,35.36)(44.98,35.21)
\qbezier(44.98,35.21)(47.43,35.05)(50,35)
\qbezier(50,35)(52.57,35.05)(55.02,35.21)
\qbezier(55.02,35.21)(57.47,35.36)(60,35)
\qbezier(60,35)(64.47,34.11)(69.52,32.17)
\qbezier(69.52,32.17)(74.57,30.23)(75,30)
\put(125,80){\makebox(0,0)[cc]{0}}

\put(155,80){\makebox(0,0)[cc]{0}}

\put(140,70){\makebox(0,0)[cc]{0}}

\put(165,55){\makebox(0,0)[cc]{0}}

\put(115,55){\makebox(0,0)[cc]{0}}

\put(155,30){\makebox(0,0)[cc]{0}}

\put(125,30){\makebox(0,0)[cc]{0}}

\put(140,40){\makebox(0,0)[cc]{0}}

\put(145,55){\makebox(0,0)[cc]{0}}

\linethickness{0.3mm}
\qbezier(115,80)(115.43,79.77)(120.48,77.83)
\qbezier(120.48,77.83)(125.53,75.89)(130,75)
\qbezier(130,75)(135.08,74.17)(140,74.17)
\qbezier(140,74.17)(144.92,74.17)(150,75)
\qbezier(150,75)(153.99,75.69)(157.63,77.04)
\qbezier(157.63,77.04)(161.26,78.4)(165,80)
\put(25,25){\makebox(0,0)[cc]{$S_B$}}

\end{picture}
  \par\end{centering}
 \caption{Examples of containment areas for SP spanning tree construction.}
\label{fig:ExSP}
\end{figure}

\begin{figure}[t]
\noindent \begin{centering} \ifx\JPicScale\undefined\def\JPicScale{0.75}\fi
\unitlength \JPicScale mm
\begin{picture}(173,95)(0,0)
\linethickness{0.3mm}
\put(50,85){\circle{10}}

\linethickness{0.3mm}
\put(30,65){\circle{10}}

\linethickness{0.3mm}
\put(70,65){\circle{10}}

\linethickness{0.3mm}
\put(30,45){\circle{10}}

\linethickness{0.3mm}
\put(70,45){\circle{10}}

\linethickness{0.3mm}
\put(50,25){\circle{10}}

\linethickness{0.3mm}
\multiput(30,70)(0.12,0.12){125}{\line(1,0){0.12}}
\linethickness{0.3mm}
\multiput(55,85)(0.12,-0.12){125}{\line(1,0){0.12}}
\linethickness{0.3mm}
\put(30,50){\line(0,1){10}}
\linethickness{0.3mm}
\multiput(30,40)(0.12,-0.12){125}{\line(1,0){0.12}}
\linethickness{0.3mm}
\multiput(55,25)(0.12,0.12){125}{\line(1,0){0.12}}
\linethickness{0.3mm}
\put(70,50){\line(0,1){10}}
\linethickness{0.3mm}
\put(35,65){\line(1,0){30}}
\linethickness{0.3mm}
\put(35,45){\line(1,0){30}}
\linethickness{0.3mm}
\multiput(35,65)(0.18,-0.12){167}{\line(1,0){0.18}}
\linethickness{0.3mm}
\put(140,85){\circle{10}}

\linethickness{0.3mm}
\put(120,65){\circle{10}}

\linethickness{0.3mm}
\put(160,65){\circle{10}}

\linethickness{0.3mm}
\put(120,45){\circle{10}}

\linethickness{0.3mm}
\put(160,45){\circle{10}}

\linethickness{0.3mm}
\put(140,25){\circle{10}}

\linethickness{0.3mm}
\multiput(120,70)(0.12,0.12){125}{\line(1,0){0.12}}
\linethickness{0.3mm}
\multiput(145,85)(0.12,-0.12){125}{\line(1,0){0.12}}
\linethickness{0.3mm}
\put(120,50){\line(0,1){10}}
\linethickness{0.3mm}
\multiput(120,40)(0.12,-0.12){125}{\line(1,0){0.12}}
\linethickness{0.3mm}
\multiput(145,25)(0.12,0.12){125}{\line(1,0){0.12}}
\linethickness{0.3mm}
\put(160,50){\line(0,1){10}}
\linethickness{0.3mm}
\put(125,65){\line(1,0){30}}
\linethickness{0.3mm}
\put(125,45){\line(1,0){30}}
\linethickness{0.3mm}
\multiput(125,65)(0.18,-0.12){167}{\line(1,0){0.18}}
\put(50,85){\makebox(0,0)[cc]{r}}

\put(140,85){\makebox(0,0)[cc]{r}}

\put(140,25){\makebox(0,0)[cc]{b}}

\put(50,25){\makebox(0,0)[cc]{b}}

\put(50,95){\makebox(0,0)[cc]{mr=10}}

\put(140,95){\makebox(0,0)[cc]{mr=10}}

\put(35,80){\makebox(0,0)[cc]{7}}

\put(65,80){\makebox(0,0)[cc]{6}}

\put(50,70){\makebox(0,0)[cc]{5}}

\put(75,55){\makebox(0,0)[cc]{4}}

\put(55,55){\makebox(0,0)[cc]{10}}

\put(25,55){\makebox(0,0)[cc]{6}}

\put(50,40){\makebox(0,0)[cc]{8}}

\put(65,30){\makebox(0,0)[cc]{32}}

\put(30,30){\makebox(0,0)[cc]{16}}

\put(140,15){\makebox(0,0)[cc]{$level_b=10$}}

\put(50,15){\makebox(0,0)[cc]{$level_b=10$}}

\put(120,30){\makebox(0,0)[cc]{11}}

\put(155,30){\makebox(0,0)[cc]{12}}

\put(155,80){\makebox(0,0)[cc]{10}}

\put(120,80){\makebox(0,0)[cc]{7}}

\put(165,55){\makebox(0,0)[cc]{13}}

\put(140,70){\makebox(0,0)[cc]{6}}

\put(145,55){\makebox(0,0)[cc]{5}}

\put(115,55){\makebox(0,0)[cc]{3}}

\put(140,40){\makebox(0,0)[cc]{1}}

\linethickness{0.3mm}
\qbezier(105,53.86)(105.43,54.17)(110.46,56.36)
\qbezier(110.46,56.36)(115.5,58.55)(120,58.86)
\qbezier(120,58.86)(124.01,58.43)(127.58,55.9)
\qbezier(127.58,55.9)(131.14,53.36)(135,53.86)
\qbezier(135,53.86)(139.98,55.48)(142.9,60.38)
\qbezier(142.9,60.38)(145.83,65.27)(150,68.86)
\qbezier(150,68.86)(152.33,70.57)(154.77,71.87)
\qbezier(154.77,71.87)(157.21,73.18)(160,73.86)
\qbezier(160,73.86)(162.53,74.5)(165.07,74.65)
\qbezier(165.07,74.65)(167.61,74.79)(170,73.86)
\put(173,68.71){\makebox(0,0)[cc]{$S_B$}}

\linethickness{0.3mm}
\qbezier(20,75)(20.66,75.25)(27.63,75.98)
\qbezier(27.63,75.98)(34.6,76.72)(40,75)
\qbezier(40,75)(41.65,74.26)(42.8,72.89)
\qbezier(42.8,72.89)(43.96,71.52)(45,70)
\qbezier(45,70)(46.51,67.55)(47.23,64.72)
\qbezier(47.23,64.72)(47.94,61.9)(50,60)
\qbezier(50,60)(54.2,56.72)(59.43,55.56)
\qbezier(59.43,55.56)(64.65,54.39)(70,55)
\qbezier(70,55)(72.83,55.38)(75.26,56.81)
\qbezier(75.26,56.81)(77.68,58.25)(80,60)
\put(15,70){\makebox(0,0)[cc]{$S_B$}}

\end{picture}
  \par\end{centering}
 \caption{Examples of containment areas for flow spanning tree construction.}
\label{fig:ExFlow}
\end{figure}
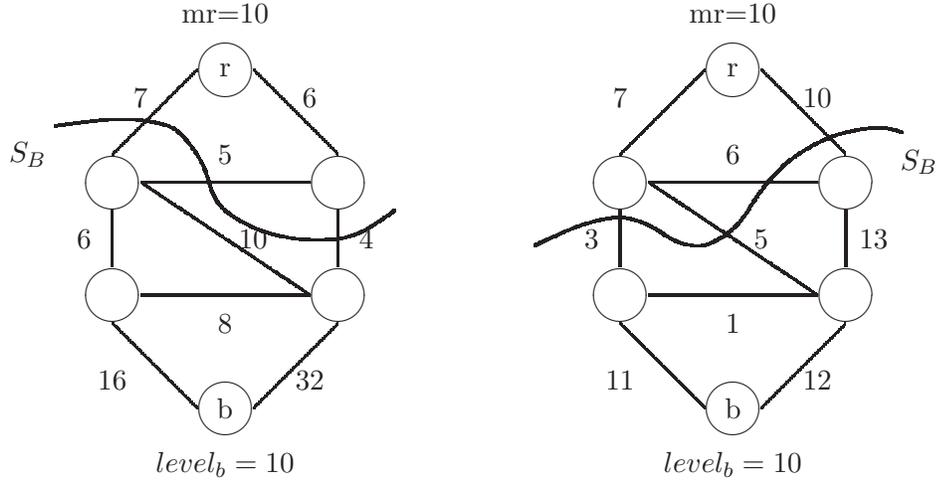

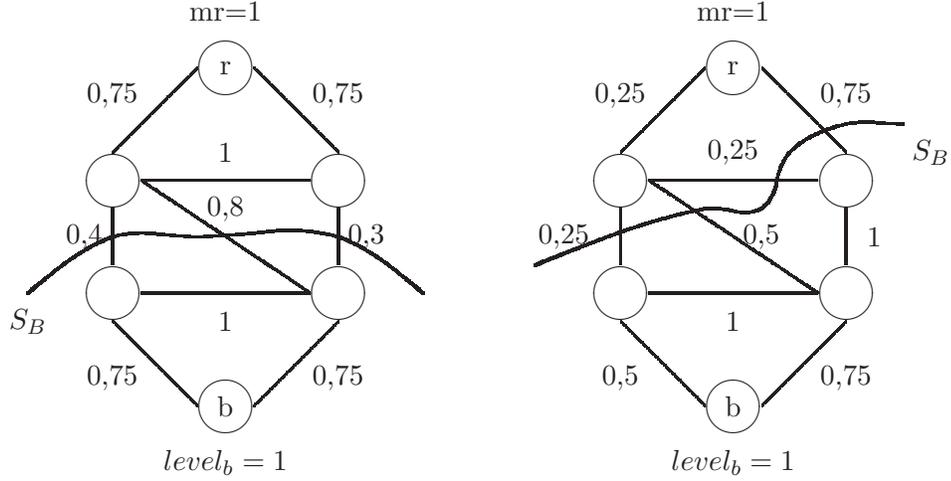
\begin{figure}[t]
\noindent \begin{centering} \ifx\JPicScale\undefined\def\JPicScale{0.75}\fi
\unitlength \JPicScale mm
\begin{picture}(175,95)(0,0)
\linethickness{0.3mm}
\put(50,85){\circle{10}}

\linethickness{0.3mm}
\put(30,65){\circle{10}}

\linethickness{0.3mm}
\put(70,65){\circle{10}}

\linethickness{0.3mm}
\put(30,45){\circle{10}}

\linethickness{0.3mm}
\put(70,45){\circle{10}}

\linethickness{0.3mm}
\put(50,25){\circle{10}}

\linethickness{0.3mm}
\multiput(30,70)(0.12,0.12){125}{\line(1,0){0.12}}
\linethickness{0.3mm}
\multiput(55,85)(0.12,-0.12){125}{\line(1,0){0.12}}
\linethickness{0.3mm}
\put(30,50){\line(0,1){10}}
\linethickness{0.3mm}
\multiput(30,40)(0.12,-0.12){125}{\line(1,0){0.12}}
\linethickness{0.3mm}
\multiput(55,25)(0.12,0.12){125}{\line(1,0){0.12}}
\linethickness{0.3mm}
\put(70,50){\line(0,1){10}}
\linethickness{0.3mm}
\put(35,65){\line(1,0){30}}
\linethickness{0.3mm}
\put(35,45){\line(1,0){30}}
\linethickness{0.3mm}
\multiput(35,65)(0.18,-0.12){167}{\line(1,0){0.18}}
\linethickness{0.3mm}
\put(140,85){\circle{10}}

\linethickness{0.3mm}
\put(120,65){\circle{10}}

\linethickness{0.3mm}
\put(160,65){\circle{10}}

\linethickness{0.3mm}
\put(120,45){\circle{10}}

\linethickness{0.3mm}
\put(160,45){\circle{10}}

\linethickness{0.3mm}
\put(140,25){\circle{10}}

\linethickness{0.3mm}
\multiput(120,70)(0.12,0.12){125}{\line(1,0){0.12}}
\linethickness{0.3mm}
\multiput(145,85)(0.12,-0.12){125}{\line(1,0){0.12}}
\linethickness{0.3mm}
\put(120,50){\line(0,1){10}}
\linethickness{0.3mm}
\multiput(120,40)(0.12,-0.12){125}{\line(1,0){0.12}}
\linethickness{0.3mm}
\multiput(145,25)(0.12,0.12){125}{\line(1,0){0.12}}
\linethickness{0.3mm}
\put(160,50){\line(0,1){10}}
\linethickness{0.3mm}
\put(125,65){\line(1,0){30}}
\linethickness{0.3mm}
\put(125,45){\line(1,0){30}}
\linethickness{0.3mm}
\multiput(125,65)(0.18,-0.12){167}{\line(1,0){0.18}}
\put(50,85){\makebox(0,0)[cc]{r}}

\put(140,85){\makebox(0,0)[cc]{r}}

\put(140,25){\makebox(0,0)[cc]{b}}

\put(50,25){\makebox(0,0)[cc]{b}}

\put(175,70){\makebox(0,0)[cc]{$S_B$}}

\put(50,95){\makebox(0,0)[cc]{mr=1}}

\put(50,15){\makebox(0,0)[cc]{$level_b=1$}}

\put(140,95){\makebox(0,0)[cc]{mr=1}}

\put(140,15){\makebox(0,0)[cc]{$level_b=1$}}

\put(70,80){\makebox(0,0)[cc]{0,75}}

\put(30,80){\makebox(0,0)[cc]{0,75}}

\put(70,30){\makebox(0,0)[cc]{0,75}}

\put(30,30){\makebox(0,0)[cc]{0,75}}

\put(50,40){\makebox(0,0)[cc]{1}}

\put(50,70){\makebox(0,0)[cc]{1}}

\put(50,60){\makebox(0,0)[cc]{0,8}}

\put(25,55){\makebox(0,0)[cc]{0,4}}

\put(75,55){\makebox(0,0)[cc]{0,3}}

\linethickness{0.3mm}
\qbezier(15,45)(15.38,45.46)(20.28,49.33)
\qbezier(20.28,49.33)(25.17,53.2)(30,55)
\qbezier(30,55)(33.69,55.98)(37.41,55.56)
\qbezier(37.41,55.56)(41.13,55.14)(45,55)
\qbezier(45,55)(51.45,55.23)(57.65,55.93)
\qbezier(57.65,55.93)(63.85,56.63)(70,55)
\qbezier(70,55)(74.83,53.2)(79.72,49.33)
\qbezier(79.72,49.33)(84.62,45.46)(85,45)
\put(120,80){\makebox(0,0)[cc]{0,25}}

\put(140,70){\makebox(0,0)[cc]{0,25}}

\put(160,80){\makebox(0,0)[cc]{0,75}}

\put(165,55){\makebox(0,0)[cc]{1}}

\put(145,55){\makebox(0,0)[cc]{0,5}}

\put(140,40){\makebox(0,0)[cc]{1}}

\put(110,55){\makebox(0,0)[cc]{0,25}}

\put(160,30){\makebox(0,0)[cc]{0,75}}

\put(120,30){\makebox(0,0)[cc]{0,5}}

\linethickness{0.3mm}
\qbezier(105,50)(105.87,50.46)(115.97,54.35)
\qbezier(115.97,54.35)(126.07,58.23)(135,60)
\qbezier(135,60)(137.59,60.19)(140.19,59.54)
\qbezier(140.19,59.54)(142.8,58.88)(145,60)
\qbezier(145,60)(147.18,61.68)(147.69,64.73)
\qbezier(147.69,64.73)(148.21,67.78)(150,70)
\qbezier(150,70)(152.05,72.01)(154.61,73.19)
\qbezier(154.61,73.19)(157.17,74.36)(160,75)
\qbezier(160,75)(162.5,75.52)(164.96,75.3)
\qbezier(164.96,75.3)(167.42,75.07)(170,75)
\put(15,40){\makebox(0,0)[cc]{$S_B$}}

\end{picture}
  \par\end{centering}
 \caption{Examples of containment areas for reliability spanning tree construction.}
\label{fig:ExReliability}
\end{figure}

Given a maximizable metric $\mathcal{M}=(M,W,mr,met,\prec)$, the aim of this work is to construct a maximum metric tree with respect to $\mathcal{M}$ which spans the system in a self-stabilizing way in a system subject to permanent Byzantine failures. It is obvious that these Byzantine processes may disturb some correct processes. It is why, we relax the problem in the following way: we want to construct a maximum metric forest with respect to $\mathcal{M}$. The root of any tree of this forest must be either the real root or a Byzantine process. 

Each process $v$ has three O-variables: a pointer to its parent in its tree ($prnt_v\in N_v\cup\{\bot\}$), a level which stores its current metric value ($level_v\in M$), and a variable which stores its distance to the root of its tree ($dist_v\in\{0,\ldots,D\}$). Obviously, Byzantine process may disturb (at least) their neighbors. We use the following specification of the problem.

We introduce new notations as follows. Given an assigned metric $(M,W,met,mr,\prec,wf)$ over the system $S$ and two processes $u$ and $v$, we denote by $\mu(u,v)$ the maximum metric of node $u$ when $v$ plays the role of the root of the system and by $w_{u,v}$ the weight of the edge $\{u,v\}$ (that is, the value of $wf(\{u,v\})$).

\begin{definition}[$\mathcal{M}$-path]
Given an assigned metric $\mathcal{M}=(M,W,mr,met,\prec,wf)$ over a system $S$, a path $(v_0,\ldots,v_k)$ ($k\geq 1$) of $S$ is a \emph{$\mathcal{M}$-path} if and only if:
\begin{enumerate}
\item $prnt_{v_0}=\bot$, $level_{v_0}=0$, $dist_{v_0}=0$, and $v_0\in B\cup\{r\}$,
\item $\forall i\in\{1,\ldots,k\}, prnt_{v_i}=v_{i-1}$, $level_{v_i}=met(level_{v_{i-1}},w_{v_i,v_{i-1}})$, and $dist_{v_i}=i$,
\item $\forall i\in\{1,\ldots,k\}, met(level_{v_{i-1}},w_{v_i,v_{i-1}})=\underset{u\in N_v}{max_\prec}\{met(level_{u},w_{v_i,u})\}$, and
\item $level_{v_{k}}=\mu(v_k,v_0)$.
\end{enumerate}
\end{definition}

We define the specification predicate $spec(v)$ of the maximum metric tree construction with respect to a maximizable metric $\mathcal{M}$ as follows.
\[spec(v) : \begin{cases}
 prnt_v = \bot,  level_v = 0 \text{, and } dist_v=0 \text{ if } v \text{ is the root } r \\
 \text{there exists a }\mathcal{M}\text{-path } (v_0,\ldots,v_k) \text{ such that } v_k=v \text{ otherwise}
\end{cases}\]

Following discussion of Section \ref{sec:stab}, it is obvious that there exists no strictly stabilizing protocol for this problem. It is why we consider the weaker notion of topology-aware strict stabilization. First, we show an impossibility result in order to define the best possible containment area. Then, we provide a maximum metric tree construction protocol which is $(S_{B},f)$-TA-strictly stabilizing where $f\leq n-1$ which match these optimal containment area, namely:

\[S_{B}=\left\{v\in V\setminus B\left|\mu(v,r)\preceq max_\prec\{\mu(v,b),b\in B\}\right.\right\}\setminus\{r\}\]

Figures from \ref{fig:ExSP} to \ref{fig:ExReliability} provide some examples of containment areas with respect to several maximizable metrics.

We introduce here a new definition that is used in the following.

\begin{definition}[Fixed point]
A metric value $m$ is a \emph{fixed point} of a metric $\mathcal{M}=(M,W,mr,met,$ $\prec)$ if $m\in M$ and if for any value $w\in W$, we have: $met(m,w)=m$.
\end{definition}

\subsection{Impossibility Result}

In this section, we show that there exists some constraints on the containment area of any topology-aware strictly stabilizing for the maximum metric tree construction depending on the metric.

\begin{theorem}\label{th:impTAstrict}
Given a maximizable metric $\mathcal{M}=(M,W,mr,met,\prec)$, even under the central daemon, there exists no $(A_B,1)$-TA-strictly stabilizing protocol for maximum metric spanning tree construction with respect to $\mathcal{M}$ where $A_B\varsubsetneq S_B$.
\end{theorem}

\begin{proof}
Let $\mathcal{M}=(M,W,mr,met,\prec)$ be a maximizable metric and $\mathcal{P}$ be a $(A_B,1)$-TA-strictly stabilizing protocol for maximum metric spanning tree construction protocol with respect to $\mathcal{M}$ where $A_B\varsubsetneq S_B$. We must distinguish the following cases:

\begin{description}
\item[Case 1:] $|M|=1$.\\
Denote by $m$ the metric value such that $M=\{m\}$. For any system and for any process $v\neq r$, we have $\mu(v,r)=\underset{b\in B}{min_\prec}\{\mu(v,b)\}=m$. Consequently, $S_B=V\setminus(B\cup\{r\})$ for any system.

Consider the following system: $V=\{r,u,v,b\}$ and $E=\{\{r,u\},\{u,v\},\{v,b\}\}$ ($b$ is a Byzantine process). As $S_B=\{u,v\}$ and $A_B\varsubsetneq S_B$, we have: $u\notin A_B$ or $v\notin A_B$. Consider now the following configuration $\rho_0^0$: $prnt_r=prnt_b=\bot$, $prnt_v=b$, $prnt_u=v$, $level_r=level_u=level_v=level_b=m$, $dist_r=dist_b=0$, $dist_v=1$ and $dist_u=2$ (see Figure \ref{fig:impTAstrict}, other variables may have arbitrary values). Note that $\rho_0^0$ is $A_B$-legitimate for $spec$ (whatever $A_B$ is).

Assume now that $b$ behaves as a correct process with respect to $\mathcal{P}$. Then, by convergence of $\mathcal{P}$ in a fault-free system starting from $\rho_0^0$ which is not legitimate (remember that a strictly-stabilizing protocol is a special case of self-stabilizing protocol), we can deduce that the system reaches in a finite time a configuration $\rho_1^0$ (see Figure \ref{fig:impTAstrict}) in which: $prnt_r=\bot$, $prnt_u=r$, $prnt_v=u$, $prnt_b=v$, $level_r=level_u=level_v=level_b=m$, $dist_r=0$, $dist_u=1$, $dist_v=2$ and $dist_b=3$. Note that processes $u$ and $v$ modify their O-variables in this execution. This contradicts the $(A_B,1)$-TA-strict stabilization of $\mathcal{P}$ (whatever $A_B$ is).

\begin{figure}[t]
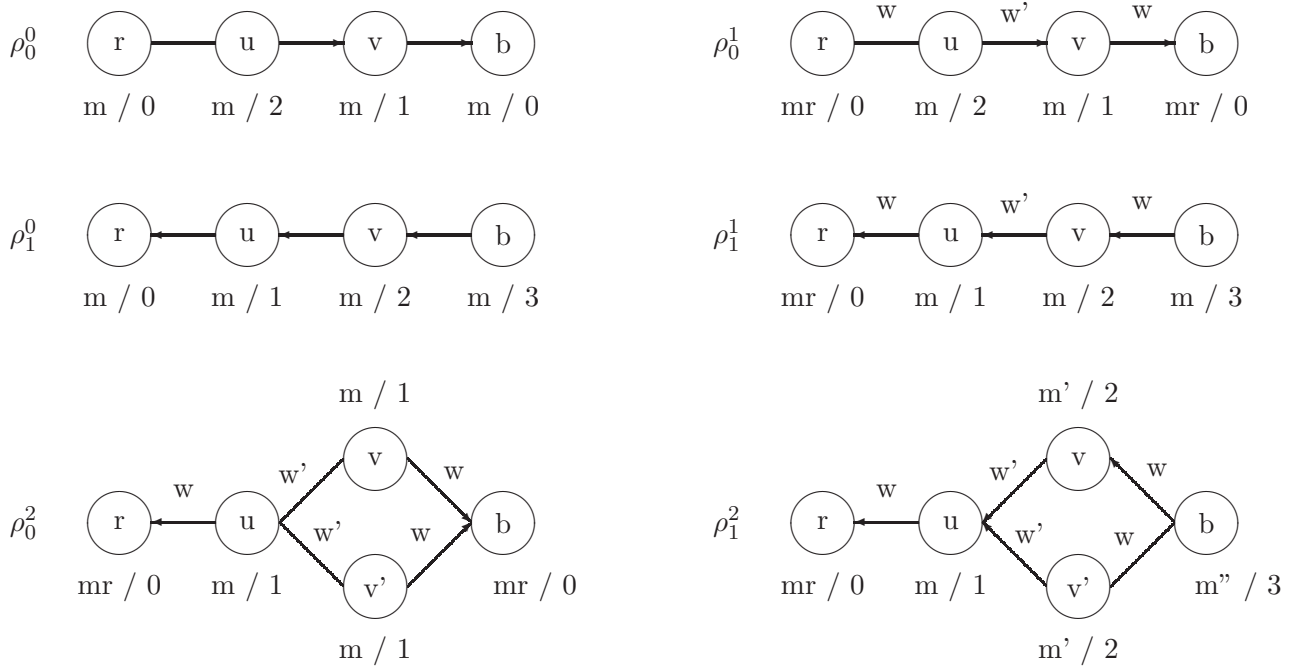

\noindent \begin{centering} \include{impTAstrict}
  \par\end{centering}
 \caption{Configurations used in proof of Theorem \ref{th:impTAstrict}.}
\label{fig:impTAstrict}
\end{figure}

\item[Case 2:] $|M|\geq 2$.\\
By definition of a bounded metric, we can deduce that there exist $m\in M$ and $w\in W$ such that $m=met(mr,w)\prec mr$. Then, we must distinguish the following cases:
\begin{description}
\item[Case 2.1:] $m$ is a fixed point of $\mathcal{M}$.\\
Consider the following system: $V=\{r,u,v,b\}$, $E=\{\{r,u\},\{u,v\},\{v,b\}\}$, $w_{r,u}=w_{v,b}=w$, and $w_{u,v}=w'$ ($b$ is a Byzantine process). As for any $w'\in W$, $met(m,w')=m$ (by definition of a fixed point), we have: $S_B=\{u,v\}$. Since $A_B\varsubsetneq S_B$, we have: $u\notin A_B$ or $v\notin A_B$. Consider now the following configuration $\rho_0^1$: $prnt_r=prnt_b=\bot$, $prnt_v=b$, $prnt_u=v$, $level_r=level_b=mr$, $level_u=level_v=m$, $dist_r=dist_b=0$, $dist_v=1$ and $dist_u=2$ (see Figure \ref{fig:impTAstrict}, other variables may have arbitrary values). Note that $\rho_0^1$ is $A_B$-legitimate for $spec$ (whatever $A_B$ is).

Assume now that $b$ behaves as a correct process with respect to $\mathcal{P}$. Then, by convergence of $\mathcal{P}$ in a fault-free system starting from $\rho_0^1$ which is not legitimate (remember that a strictly-stabilizing protocol is a special case of self-stabilizing protocol), we can deduce that the system reaches in a finite time a configuration $\rho_1^1$ (see Figure \ref{fig:impTAstrict}) in which: $prnt_r=\bot$, $prnt_u=r$, $prnt_v=u$, $prnt_b=v$, $level_r=mr$, $level_u=level_v=level_b=m$ (since $m$ is a fixed point), $dist_r=0$, $dist_u=1$, $dist_v=2$ and $dist_b=3$. Note that processes $u$ and $v$ modify their O-variables in this execution. This contradicts the $(A_B,1)$-TA-strict stabilization of $\mathcal{P}$ (whatever $A_B$ is).
\item[Case 2.2:] $m$ is not a fixed point of $\mathcal{M}$.\\
This implies that there exists $w'\in W$ such that: $met(m,w')\prec m$ (remember that $\mathcal{M}$ is bounded). Consider the following system: $V=\{r,u,v,v',b\}$, $E=\{\{r,u\},\{u,v\},\{u,v'\},$ $\{v,b\},\{v',b\}\}$, $w_{r,u}=w_{v,b}=w_{v',b}=w$, and $w_{u,v}=w_{u,v'}=w'$ ($b$ is a Byzantine process). We can see that $S_B=\{v,v'\}$. Since $A_B\varsubsetneq S_B$, we have: $v\notin A_B$ or $v'\notin A_B$. Consider now the following configuration $\rho_0^2$: $prnt_r=prnt_b=\bot$, $prnt_v=prnt_{v'}=b$, $prnt_u=r$, $level_r=level_b=mr$, $level_u=level_v=level_{v'}=m$, $dist_r=dist_b=0$, $dist_v=dist_{v'}=1$ and $dist_u=1$ (see Figure \ref{fig:impTAstrict}, other variables may have arbitrary values). Note that $\rho_0^2$ is $A_B$-legitimate for $spec$ (whatever $A_B$ is).

Assume now that $b$ behaves as a correct process with respect to $\mathcal{P}$. Then, by convergence of $\mathcal{P}$ in a fault-free system starting from $\rho_0^2$ which is not legitimate (remember that a strictly-stabilizing protocol is a special case of self-stabilizing protocol), we can deduce that the system reaches in a finite time a configuration $\rho_1^2$ (see Figure \ref{fig:impTAstrict}) in which: $prnt_r=\bot$, $prnt_u=r$, $prnt_v=prnt_{v'}=u$, $prnt_b=v$ (or $prnt_b=v'$), $level_r=mr$, $level_u=m$ $level_v=level_{v'}=met(m,w')=m'$, $level_b=met(m',w)=m''$, $dist_r=0$, $dist_u=1$, $dist_v=dist_{v'}=2$ and $dist_b=3$. Note that processes $v$ and $v'$ modify their O-variables in this execution. This contradicts the $(A_B,1)$-TA-strict stabilization of $\mathcal{P}$ (whatever $A_B$ is).
\end{description}
\end{description}
\end{proof}

\subsection{Topology-Aware Strict Stabilizing Protocol}

In this section, we provide our self-stabilizing protocol that achieve optimal containment areas to permanent Byzantine failures for constructing a maximum metric tree for any maximizable metric $\mathcal{M}=(M,W,met,mr,\prec)$. More formally, our protocol is $(S_B,f)$-strictly stabilizing, that is optimal with respect to the result of Theorem \ref{th:impTAstrict}. Our protocol is borrowed from the one of \cite{GS99c} (which is self-stabilizing). The key idea of this protocol is to use the distance variable (upper bounded by a given constant $D$) to detect and break cycles of process which has the same maximum metric. The main modification we bring to this protocol follows. In the initial protocol, when a process modifies its parent, it chooses arbitrarily one of the "better" neighbors (with respect to the metric).  To achieve the $(S_B,f)$-TA-strict stabilization, we must ensures a fair selection along the set of its neighbor. We perform this fairness with a round-robin order along the set of neighbors. Our solution is presented as Algorithm~\ref{algo:max}.

\begin{algorithm}
\caption{$\mathcal{SSMAX}$: A TA-strictly stabilizing protocol for maximum metric tree construction.}\label{algo:max}
\scriptsize
\begin{description}
\item{Data:}~\\
$N_v$: totally ordered set of neighbors of $v$.\\
$D$: upper bound of the number of processes in a simple path.
\item{Variables:}~\\
$prnt_v\begin{cases}=\bot \text{ if } v=r\\\in N_v \text{ if } v\neq r\end{cases}$: pointer on the parent of $v$ in the tree.\\
$level_v\in\{m\in M|m\preceq mr\}$: metric of the node.\\
$dist_v\in\{0,\ldots,D\}$: distance to the root.
\item{Macro:}~\\
For any subset $A\subseteq N_v$, $choose(A)$ returns the first element of $A$ which is bigger than $prnt_v$ (in a round-robin fashion).
\item{Rules:}~\\
$\boldsymbol{(R_r)}::(v=r)\wedge((level_v\neq mr)\vee(dist_v\neq 0))\longrightarrow level_v:=mr;~dist_v:=0$

$\boldsymbol{(R_1)}::(v\neq r)\wedge(prnt_v\in N_v)\wedge((dist_v\neq min(dist_{prnt_v}+1,D))\vee(level_v\neq met(level_{prnt_v},w_{v,prnt_v})))$\\ $~~~~~~~~~~~~~\longrightarrow dist_v:=min(dist_{prnt_v}+1,D);level_v:=met(level_{prnt_v},w_{v,prnt_v})$

$\boldsymbol{(R_2)}::(v\neq r)\wedge(dist_v=D)\wedge(\exists u \in N_v, dist_u<D-1)$\\$~~~~~~~~~~~~~\longrightarrow prnt_v:=choose(\{u\in N_v|dist_v<D-1\});~dist_v:=dist_{prnt_v}+1;~level_v:=met(level_{prnt_v},w_{v,prnt_v})$

$\boldsymbol{(R_3)}::(v\neq r)\wedge(\exists u\in N_v,(dist_u<D-1)\wedge(level_v\prec met(level_u,w_{u,v})))$\\
$~~~~~~~~~~~~~\longrightarrow prnt_v:=choose\Bigg(\Bigg\{u\in N_v\Big|(level_u<D-1)\wedge(met(level_u,w_{u,v})=\underset{q\in N_v/level_q<D-1}{max_\prec}\{met(level_q,w_{q,v})\})\Bigg\}\Bigg);$\\
$~~~~~~~~~~~~~~~~~level_v:=met(level_{prnt_v},w_{prnt_v,v});~dist_v:=dist_{prnt_v}+1$
\end{description}
\end{algorithm}
\normalsize

In the following, we provide the proof of the TA-strict stabilization of $\mathcal{SSMAX}$. Remember that the real root $r$ can not be a Byzantine process by hypothesis. Note that the subsystem whose set of nodes is $V\setminus S_B$ is connected respectively by boundedness of the metric.

\begin{lemma}\label{lem:goodProperty}
For any process $v\in V$, we have:
\[\forall u\in N_v,met\left(\underset{p\in B\cup\{r\}}{max_\prec}\{\mu(u,p)\},w_{u,v}\right)\preceq\underset{p\in B\cup\{r\}}{max_\prec}\{\mu(v,p)\}\]
\end{lemma}

\begin{proof}
Let $v\in V$ be a process. By contradiction, assume that there exists a neighbor $u$ of $v$ such that:
\[\underset{p\in B\cup\{r\}}{max_\prec}\{\mu(v,p)\}\prec met\left(\underset{p\in B\cup\{r\}}{max_\prec}\{\mu(u,p)\},w_{u,v}\right)\]
Let $q\in B\cup\{r\}$ one of the process such that $\underset{p\in B\cup\{r\}}{max_\prec}\{\mu(u,p)\}=\mu(u,q)$. Then, we have:
\[\begin{array}{rcll}
\underset{p\in B\cup\{r\}}{max_\prec}\{\mu(v,p)\}&\prec&met(\mu(u,q),w_{u,v})&\text{ by construction of }q\\
&\prec&\mu(v,q)&\text{ since } met(\mu(u,q),w_{u,v})\preceq \mu(v,q)\\
\end{array}\]
This contradicts the fact that $q\in B\cup\{r\}$ and shows us the result.
\end{proof}

Given a configuration $\rho\in C$ and a metric value $m\in M$, let us define the following predicate: 
\[IM_m(\rho)\equiv \forall v\in V,level_v\preceq max_\prec\left\{m,\underset{u\in B\cup\{r\}}{max_\prec}\{\mu(v,u)\}\right\}\]

\begin{lemma}\label{lem:Imclosed}
For any metric value $m\in M$, the predicate $IM_m$ is closed by actions of $\mathcal{SSMAX}$.
\end{lemma}

\begin{proof}
Let $m$ be a metric value ($m\in M$). Let $\rho\in C$ be a configuration such that $IM_m(\rho)=true$ and $\rho'\in C$ be a configuration such that $\rho \stackrel{R}{\mapsto} \rho'$ is a step of $\mathcal{SSMAX}$.

If the root process $r\in R$ (respectively a Byzantine process $b\in R$), then we have $level_r=mr$ (respectively $level_b\preceq mr$) in $\rho'$ by construction of $\boldsymbol{(R_r)}$ (respectively by definition of $level_b$). Hence, $level_r\preceq max_\prec\left\{m, \underset{u\in B\cup\{r\}}{max_\prec}\{\mu(r,u)\}\right\}=mr$ (respectively $level_b\preceq max_\prec\left\{m,\underset{u\in B\cup\{r\}}{max_\prec}\{\mu(b,u)\}\right\}\preceq mr$).

If a correct process $v\in R$ with $v\neq r$, then there exists a neighbor $p$ of $v$ such that $level_p\preceq max_\prec\left\{m,\underset{u\in B\cup\{r\}}{max_\prec}\{\mu(p,u)\}\right\}$ in $\rho$ (since $IM(\rho)=true$) and $prnt_v=p$ and $level_v=met(level_p,$ $w_{v,p})$ in $\rho'$ (since $v$ is activated during this step).

If we apply the Lemma \ref{lem:goodProperty} to $met$ and to neighbor $p$, we obtain the following property:
\[met\left(\underset{u\in B\cup\{r\}}{max_\prec}\{\mu(p,u)\},w_{v,p}\right)\preceq\underset{u\in B\cup\{r\}}{max_\prec}\{\mu(v,u)\}\]  

Consequently, we obtain that, in $\rho'$:
\[\begin{array}{rcll}
level_v & = & met(level_p,w_{v,p})&\\
& \preceq & met\left(max_\prec\left\{m,\underset{u\in B\cup\{r\}}{max_\prec}\{\mu(p,u)\}\right\},w_{v,p}\right)&\text{ by boundedness of }\mathcal{M}\\
& \preceq & max_\prec\left\{met(m,w_{v,p}),met\left(\underset{u\in B\cup\{r\}}{max_\prec}\{\mu(p,u)\},w_{v,p}\right) \right\}&\\
& \preceq & max_\prec\left\{m,\underset{u\in B\cup\{r\}}{max_\prec}\{\mu(v,u)\}\right\}&\text{ since } met(m,w_{v,p})\preceq m
\end{array}\]

We can deduce that $IM_d(\rho')=true$, that concludes the proof.
\end{proof}

Given an assigned metric to a system $G$, we can observe that the set of metrics value $M$ is finite and that we can label elements of $M$ by $m_0=mr,m_1,\ldots,m_k$ in a way such that $\forall i\in\{0,\ldots,k-1\},m_{i+1}\prec m_i$.

We introduce the following notations:
\[\begin{array}{rrcl}
\forall m_i\in M, & P_{m_i} & = & \big\{v\in V\setminus S_B\big| \mu(v,r)=m_i\big\}\\
\forall m_i\in M, & V_{m_i} & = & \underset{j=0}{\overset{i}{\bigcup}}P_{m_j}\\
\forall m_i\in M, & I_{m_i} & = & \big\{v\in V\big|\underset{u\in B\cup\{r\}}{max_\prec}\{\mu(v,u)\}\prec m_i \big\}\\
\forall m_i\in M, & \mathcal{LC}_{m_i} & = & \big\{\rho\in\mathcal{C}\big|(\forall v\in V_{m_i}, spec(v))\wedge(IM_{m_i}(\rho))\big\}\\
 & \mathcal{LC} & = & \mathcal{LC}_{m_k}
\end{array}\]

\begin{lemma}\label{lem:LCmiclosed}
For any $m_i\in M$, the set $\mathcal{LC}_{m_i}$ is closed by actions of $\mathcal{SSMAX}$.
\end{lemma}

\begin{proof}
Let $m_i$ be a metric value from $M$ and $\rho$ be a configuration of $\mathcal{LC}_{m_i}$. By construction, any process $v\in V_{m_i}$ satisfies $spec(v)$ in $\rho$. 

In particular, the root process satisfies: $prnt_r=\bot$, $level_r=mr$, and $dist_r=0$. By construction of $\mathcal{SSMAX}$, $r$ is not enabled and then never modifies its O-variables (since the guard of the rule of $r$ does not involve the state of its neighbors). 

In the same way, any process $v\in V_{m_i}$ satisfies: $prnt_v\in N_v$, $level_v=met(level_{prnt_v},$ $w_{prnt_v,v})$, $dist_v=dist_{prnt_v}+1$, and $level_v=\underset{u\in N_v}{max_\prec}\{met(level_u,w_{u,v})\}$. Note that, as $v\in V_{m_i}$ and $spec(v)$ holds in $\rho$, we have: $level_v=\mu(v,r)=\underset{p\in B\cup\{r\}}{max_\prec}\{\mu(v,p)\}$ and $dist_v\leq D-1$ by construction of $D$. Hence, process $v$ is not enabled in $\rho$.

Assume that there exists a process $v\in V_{m_i}$ that take a step $\rho' \stackrel{R}{\mapsto} \rho''$ in an execution starting from $\rho$ (without loss of generality, assume that $v$ is the first process of $v\in V_{m_i}$ that takes a step in this execution). Then, we know that $v\neq r$. This activation implies that a neighbor $u\notin V_{m_i}$ (since $v$ is the first process of $V_{m_i}$ to take a step) of $v$ modified its $level_u$ variable to a metric value $m\in M$ such that $level_v\prec met(m,w_{u,v})$ in $\rho'$ (note that O-variables of $v$ and $prnt_v$ remain consistent since $v$ is the first process to take a step in this execution).

Hence, we have $level_v=\underset{p\in B\cup\{r\}}{max_\prec}\{\mu(v,p)\}\prec met(m,w_{u,v})$. Moreover, the closure of $IM_B$ (established in Lemma \ref{lem:Imclosed}) ensures us that $m\preceq \underset{p\in B\cup\{r\}}{max_\prec}\{\mu(u,p)\}$. By boundedness of $\mathcal{M}$, we can deduce that $met(m,w_{u,v})\preceq met(\underset{p\in B\cup\{r\}}{max_\prec}\{\mu(u,p)\},w_{u,v})$. Consequently, we obtain that $\underset{p\in B\cup\{r\}}{max_\prec}\{\mu(v,p)\}\prec met(\underset{p\in B\cup\{r\}}{max_\prec}\{\mu(u,p)\},w_{u,v})$. This is contradictory with the result of Lemma \ref{lem:goodProperty}.

In conclusion, any process $v\in V_{m_i}$ takes no step in any execution starting from $\rho$ and then always satisfies $spec(v)$. Then, the closure of $IM_B$ (established in Lemma \ref{lem:Imclosed}) concludes the proof.
\end{proof}

\begin{lemma}\label{lem:SBTAcontainedMax}
Any configuration of $\mathcal{LC}$ is $(S_B,n-1)$-TA contained for $spec$.
\end{lemma}

\begin{proof}
This is a direct application of the Lemma \ref{lem:LCmiclosed} to $\mathcal{LC}=\mathcal{LC}_{m_i}$.
\end{proof}

\begin{lemma}\label{lem:CtoLCmr}
Starting from any configuration of $\mathcal{C}$, any execution of $\mathcal{SSMAX}$ reaches in a finite time a configuration of $\mathcal{LC}_{mr}$.
\end{lemma}

\begin{proof}
Let $\rho$ be an arbitrary configuration. Then, it is obvious that $IM_{mr}(\rho)$ is satisfied. By closure of $IM_{mr}$ (proved in Lemma \ref{lem:Imclosed}), we know that $IM_{mr}$ remains satisfied in any execution starting from $\rho$.

If $r$ does not satisfy $spec(r)$ in $\rho$, then $r$ is continuously enabled. Since the scheduling is weakly fair, $r$ is activated in a finite time and then $r$ satisfies $spec(r)$ in a finite time. Denote by $\rho'$ the first configuration in which $spec(r)$ holds. Note that $r$ takes no step in any execution starting from $\rho'$.

The boundedness of $\mathcal{M}$ implies that $P_{mr}$ induces a connected subsystem. If $P_{mr}=\{r\}$, then we proved that $\rho'\in\mathcal{LC}_{mr}$ and we have the result.

Otherwise, observe that, for any configuration of an execution starting from $\rho'$, if all processes of $P_{mr}$ are not enabled, then all processes $v$ of $P_{mr}$ satisfy $spec(v)$. Assume now that there exists an execution $e$ starting from $\rho'$ in which some processes of $P_{mr}$ takes infinitely many steps. By construction, at least one of these processes (note it $v$) has a neighbor $u$ which takes only a finite number of steps in $e$ (recall that $P_{mr}$ induces a connected subsystem and that $r$ takes no step in $e$). After $u$ takes its last step of $e$, we can observe that $level_u=mr$ and $dist_u<D-1$ (otherwise, $u$ is activated in a finite time that contradicts its construction). 

As $v$ can execute consequently $\boldsymbol{(R_1)}$ only a finite number of times (since the incrementation of $dist_v$ is bounded by $D$), we can deduce that $v$ executes $\boldsymbol{(R_2)}$ or $\boldsymbol{(R_3)}$ infinitely often. In both cases, $u$ belongs to the set which is the parameter of function $choose$. By the fairness of this function, we can deduce that $prnt_v=u$ in a finite time in $e$. Then, the construction of $u$ implies that $v$ is never enabled in the sequel of $e$. This is contradictory with the construction of $e$.

Consequently, any execution starting from $\rho'$ reaches in a finite time a configuration such that all processes of $P_{mr}$ are not enabled. We can deduce that this configuration belongs to $\mathcal{LC}_{mr}$, that ends the proof.
\end{proof}

\begin{lemma}\label{lem:LCmitodistD}
For any $m_i\in M$ and for any configuration $\rho\in \mathcal{LC}_{m_i}$, any execution of $\mathcal{SSMAX}$ starting from $\rho$ reaches in a finite time a configuration such that:
\[\forall v\in I_{m_i},level_v=m_i\Rightarrow dist_v=D\]
\end{lemma}

\begin{proof}
Let $m_i$ be an arbitrary metric value of $M$ and $\rho_0$ be an arbitrary configuration of $\mathcal{LC}_{m_i}$. Let $e=\rho_0,\rho_1,\ldots$ be an execution starting from $\rho_0$.

Note that $\rho_0$ satisfies $IM_{m_i}$ by construction. Hence, we have $\forall v\in I_{m_i},level_v\preceq m_i$. The closure  of $IM_{m_i}$ (proved in Lemma \ref{lem:Imclosed}) ensures us that this property is satisfied in any configuration of $e$. 

If any process $v\in I_{m_i}$ satisfies $level_v\prec m_i$ in $\rho_0$, then the result is obvious. Otherwise, we define the following variant function. For any configuration $\rho_j$ of $e$, we denote by $A_j$ the set of processes $v$ of $I_{m_i}$ such that $level_v=m_i$ in $\rho_j$. Then, we define $f(\rho_j)=\underset{v\in A_j}{min}\{dist_v\}$. We will prove the result by showing that there exists an integer $k$ such that $f(\rho_k)=D$.

First, if a process $v$ joins $A_j$ (that is, $v\notin A_{j-1}$ but $v\in A_j$), then it takes a distance value greater or equals to $f(\rho_{j+1})$ by construction of the protocol. We can deduce that the fact that some processes join $A_j$ does not decrease $f$. Moreover, the construction of the protocol implies that a process $v$ such that $v\in A_j$ and $v\in A_{j+1}$ can not decrease its distance value in the step $\rho_j\mapsto \rho_{j+1}$.

Then, consider for a given configuration $\rho_j$ a process $v\in A_j$ such that $dist_v=f(\rho_j)<D$. We distinguish the following cases:

\begin{description}
\item[Case 1:] $level_v=met(level_{prnt_v},w_{v,prnt_v})$\\
The fact that $v\in I_{m_i}$, the boundedness of $\mathcal{M}$ and the closure of $IM_{m_i}$ imply that $prnt_v\in A_j$ (and, hence that $level_{prnt_v}=m_i$). Then, by construction of $f(\rho_j)$, we know that $dist_v\neq dist_{prnt_v}+1$ (otherwise, we do not have $dist_v=f(\rho_j)$ since $prnt_v$ has a smaller distance value). Consequently, $v$ is enabled by $\boldsymbol{(R_1)}$ in $\rho_j$ and $dist_v$ increase of at least 1 during the step $\rho_j\mapsto \rho_{j+1}$ if this rule is executed.
\item[Case 2:] $level_v\neq met(level_{prnt_v},w_{v,prnt_v})$\\
The rule $\boldsymbol{(R_1)}$ is then enabled for $v$. If this rule is executed during the step $\rho_j\mapsto \rho_{j+1}$, one of the two following sub cases appears.
\begin{description}
\item[Case 2.1:] $met(level_{prnt_v},w_{v,prnt_v})\prec m_i$\\
Then, $v$ does not belong to $A_{j+1}$ by definition. 
\item[Case 2.2:] $met(level_{prnt_v},w_{v,prnt_v})=m_i$\\
Remind that the closure of $IM_{m_i}$ implies then that $level_{prnt_v}=m_i$. By construction of $f(\rho_j)$, we have $dist_{prnt_v}\geq f(\rho_j)$ in $\rho_j$. Then, we can see that $dist_v$ increases of at least 1 during the step $\rho_j\mapsto \rho_{j+1}$.
\end{description}
\end{description}

In all cases, $v$ is enabled by $\boldsymbol{(R_1)}$ in $\rho_j$ and the execution of this rule either increases strictly $dist_v$ or removes $v$ from $A_{j+1}$.

As $I_{m_i}$ is finite and the scheduling is weakly fair, we can deduce that $f$ increases in a finite time in any execution starting from $\rho_j$. By repeating the argument at most $D$ times, we can deduce that $e$ contains a configuration $\rho_k$ such that $f(\rho_k)=D$, that shows the result.
\end{proof}

\begin{lemma}\label{lem:distDtolevelmi}
For any $m_i\in M$ and for any configuration $\rho\in \mathcal{LC}_{m_i}$ such that
$\forall v\in I_{m_i},level_v=m_i\Rightarrow dist_v=D$, any execution of $\mathcal{SSMAX}$ starting from $\rho$ reaches in a finite time a configuration such that:
\[\forall v\in I_{m_i},level_v\prec m_i\]
\end{lemma}

\begin{proof}
Let $m_i\in M$ be an arbitrary metric value and $\rho_0$ be a configuration of $\mathcal{LC}_{m_i}$ such that $\forall v\in I_{m_i},level_v=m_i\Rightarrow dist_v=D$. Let $e=\rho_0,\rho_1,\ldots$ be an arbitrary execution starting from $\rho_0$.

For any configuration $\rho_j$ of $e$, let us denote $E_{\rho_j}=\{v\in I_{m_i}|level_v=m_i\}$. By the closure of $IM_{m_i}$ (which holds by definition in $\rho_0$) established in Lemma \ref{lem:Imclosed}, we obtain the result if there exists a configuration $\rho_j$ of $e$ such that $E_{\rho_j}=\emptyset$.

If there exists some processes $v\in I_{m_i}\setminus E_{\rho_0}$ (and hence $level_v\prec m_i$) such that $prnt_v\in E_{\rho_0}$ and $met(level_{prnt_v},w_{v,prnt_v})=m_i$ in $\rho_0$, then we can observe that these processes are continuously enabled by $\boldsymbol{(R_1)}$. As the scheduling is weakly fair, $v$ activates this rule in a finite time and then, $level_v=m_i$ and $dist_v=D$. In other words, $v$ joins $E_{\rho_l}$ for a given integer $l$. We can conclude that there exists an integer $k$ such that for any $v\in I_{m_i}\setminus E_{\rho_0}$, either $prnt_v\notin E_{\rho_k}$ or $met(level_{prnt_v},w_{v,prnt_v})\prec m_i$.

Then, we prove that, for any integer $j\geq k$, we have $E_{\rho_{j+1}}\subseteq E_{\rho_j}$. For the sake of contradiction, assume that there exists an integer $j\geq k$ and a process $v\in I_{m_i}$ such that $v\in E_{\rho_{j+1}}$ and $v\notin E_{\rho_j}$. Without loss of generality, assume that $j$ is the smallest integer which performs these properties. Let us study the following cases:
\begin{description}
\item[Case 1:] If $v$ activates $\boldsymbol{(R_1)}$ during the step $\rho_j\mapsto \rho_{j+1}$, then we know that $prnt_v\notin E_{\rho_j}$ in $\rho_j$ (otherwise, we have a contradiction with the fact that $v\in E_{\rho_{j+1}}$). But in this case, we have: $level_{prnt_v}\prec m_i$. The boundedness of $\mathcal{M}$ implies that $level_v\prec m_i$ in $\rho_{j+1}$ that contradicts the fact that $v\in E_{\rho_{j+1}}$.
\item[Case 2:] If $v$ activates either $\boldsymbol{(R_2)}$ or $\boldsymbol{(R_3)}$ during the step $\rho_j\mapsto \rho_{j+1}$, then $v$ chooses a new parent which has a distance smaller than $D-1$ in $\rho_j$. This implies that this new parent does not belongs to $E_{\rho_j}$. Then, we have $level_{prnt_v}\prec m_i$. The boundedness of $\mathcal{M}$ implies that $level_v\prec m_i$ in $\rho_{j+1}$ that contradicts the fact that $v\in E_{\rho_{j+1}}$.
\end{description}
In the two cases, our claim is satisfied. In other words, there exists a point of the execution afterwards the set $E$ can not grow (this implies that, if a process leave the set $E$, it is a definitive leaving).

Assume now that there exists a step $\rho_j\mapsto \rho_{j+1}$ (with $j\geq k$) such that a process $v\in E_{\rho_j}$ is activated. Observe that the closure of $IM_{m_i}$ implies that $v$ can not be activated by the rule $\boldsymbol{(R_3)}$. If $v$ activates $\boldsymbol{(R_1)}$ during this step, then $v$ modifies its level during this step (otherwise, we have a contradiction with the fact that $level_{prnt_v}=m_i\Rightarrow dist_v=D$). The closure of $IM_{m_i}$ implies that $v$ leaves the set $E$ during this step. If $v$ activates $\boldsymbol{(R_2)}$ during this step, then $v$ chooses a new parent which has a distance smaller than $D-1$ in $\rho_j$. This implies that this new parent does not belongs to $E_{\rho_j}$. Then, we have $level_{prnt_v}\prec m_i$. The boundedness of $\mathcal{M}$ implies that $level_v\prec m_i$ in $\rho_{j+1}$. In other words, if a process of $E_{\rho_j}$ is activated during the step $\rho_j\mapsto \rho_{j+1}$, then it satisfies $v\notin E_{\rho_{j+1}}$.

Finally, observe that the construction of the protocol and the construction of the bound $D$ ensures us that any process $v\in I_{m_i}$ such that $dist_v=D$ is activated in a finite time. In conclusion, we obtain that there exists an integer $j$ such that $E_{\rho_j}=\emptyset$, that implies the result.
\end{proof}

\begin{lemma}\label{lem:CmitoIMmi+1}
For any $m_i\in M$ and for any configuration $\rho\in \mathcal{LC}_{m_i}$, any execution of $\mathcal{SSMAX}$ starting from $\rho$ reaches in a finite time a configuration $\rho'$ such that $IM_{m_{i+1}}$ holds.
\end{lemma}

\begin{proof}
This result is a direct consequence of Lemmas \ref{lem:LCmitodistD} and \ref{lem:distDtolevelmi}.
\end{proof}

\begin{lemma}\label{lem:LCmitoLCmi+1}
For any $m_i\in M$ and for any configuration $\rho\in \mathcal{LC}_{m_i}$, any execution of $\mathcal{SSMAX}$ starting from $\rho$ reaches in a finite time a configuration of $\mathcal{LC}_{m_{i+1}}$.
\end{lemma}

\begin{proof}
Let $m_i$ be a metric value of $M$ and $\rho$ be an arbitrary configuration of $\mathcal{LC}_{m_i}$. We know by Lemma \ref{lem:CmitoIMmi+1} that any execution starting from $\rho$ reaches in a finite time a configuration $\rho'$ such that $IM_{m_{i+1}}$ holds. By closure of $IM$ and of $\mathcal{LC}_{m_i}$ (established respectively in Lemma \ref{lem:Imclosed} and \ref{lem:LCmiclosed}), we know that any configuration of any execution starting from $\rho'$ belongs to $\mathcal{LC}_{m_i}$ and satisfies $IM_{m_{i+1}}$.

We know that $V_{m_i}\neq\emptyset$ since $r\in V_{m_i}$ for any $i\geq 0$. Remind that $V_{m_{i+1}}$ is connected by the boundedness of $\mathcal{M}$. Then, we know that there exists at least one process $p$ of $P_{m_{i+1}}$ which has a neighbor $q$ in $V_{m_{i}}$ such that $\mu(p,r)=met(\mu(q,r),w_{p,q})$. Moreover, Lemma \ref{lem:LCmiclosed} ensures us that any process of $V_{m_i}$ takes no step in any executions tarting from $\rho'$.

Observe that, for any configuration of an execution starting from $\rho'$, if all processes of $P_{m_{i+1}}$ are not enabled, then all processes $v$ of $P_{m_{i+1}}$ satisfy $spec(v)$. Assume now that there exists an execution $e$ starting from $\rho'$ in which some processes of $P_{m_{i+1}}$ take infinitely many steps. By construction, at least one of these processes (note it $v$) has a neighbor $u$ such that $\mu(v,r)=met(\mu(u,r),w_{v,u})$ which takes only a finite number of steps in $e$ (recall the construction of $p$). After $u$ takes its last step of $e$, we can observe that $level_u=\mu(u,r)$ and $dist_u<D-1$ (otherwise, $u$ is activated in a finite time that contradicts its construction). 

As $v$ can execute consequently $\boldsymbol{(R_1)}$ only a finite number of times (since the incrementation of $dist_v$ is bounded by $D$), we can deduce that $v$ executes $\boldsymbol{(R_2)}$ or $\boldsymbol{(R_3)}$ infinitely often. In both cases, $u$ belongs to the set which is the parameter of function $choose$ (remind that $IM_{m_{i+1}}$ is satisfied and that $u$ has the better possible metric along $v$'s neighbors). By the construction of this function, we can deduce that $prnt_v=u$ in a finite time in $e$. Then, the construction of $u$ implies that $v$ is never enabled in the sequel of $e$. This is contradictory with the construction of $e$.

Consequently, any execution starting from $\rho'$ reaches in a finite time a configuration such that all processes of $P_{m_{i+1}}$ are not enabled. We can deduce that this configuration belongs to $\mathcal{LC}_{m_{i+1}}$, that ends the proof.
\end{proof}

\begin{lemma}\label{lem:convergenceLCMax}
Starting from any configuration, any execution of $\mathcal{SSMAX}$ reaches a configuration of $\mathcal{LC}$ in a finite time.
\end{lemma}

\begin{proof}
Let $\rho$ be an arbitrary configuration. We know by Lemma \ref{lem:CtoLCmr} that any execution starting from $\rho$ reaches in a finite time a configuration of $\mathcal{LC}_{mr}=\mathcal{LC}_{m_0}$. Then, we can apply at most $k$ times the result of Lemma \ref{lem:LCmitoLCmi+1} to obtain that any execution starting from $\rho$ reaches in a finite time a configuration of $\mathcal{LC}_{m_k}=\mathcal{LC}$, that proves the result.
\end{proof}

\begin{theorem}\label{th:SSMAXstrict}
$\mathcal{SSMAX}$ is a $(S_B,n-1)$-TA-strictly stabilizing protocol for $spec$.
\end{theorem}

\begin{proof}
This result is a direct consequence of Lemmas \ref{lem:SBTAcontainedMax} and \ref{lem:convergenceLCMax}.
\end{proof}

Note that Theorem \ref{th:impTAstrict} ensures us that $S_B$ is the optimal containment area for a topology-aware strictly stabilizing protocol for $spec$.

\section{Conclusion}

We introduced a new notion of Byzantine containment in self-stabilization: the topology-aware strict stabilization. This notion relaxes the constraint on the containment radius of the strict stabilization to a containment area. In other words, the set of correct processes which may be infinitely often disturbed by Byzantine processes is a function depending on the topology of the system and on the actual location of Byzantine processes. We illustrated the relevance of this notion by providing a topology-aware strictly stabilizing protocol for the maximum metric tree construction problem which does not admit strictly stabilizing solution. Moreover, our protocol performs the optimal containment area with respect to the topology-aware strict stabilization.

Our work raises some opening questions. Number of problems do not accept strictly stabilizing solution. Does any of them admit a topology-aware strictly stabilizing solution ? Is it possible to give a necessary and/or sufficient condition for a problem to admit a topology-aware strictly stabilizing solution ? What happens if we consider only bounded Byzantine behavior ?

\bibliographystyle{plain}
\bibliography{biblio}

\begin{thebibliography}{10}

\bibitem{BDH08c}
Michael Ben-Or, Danny Dolev, and Ezra~N. Hoch.
\newblock Fast self-stabilizing byzantine tolerant digital clock
  synchronization.
\newblock In Rida~A. Bazzi and Boaz Patt-Shamir, editors, {\em PODC}, pages
  385--394. ACM, 2008.

\bibitem{DD05c}
Ariel Daliot and Danny Dolev.
\newblock Self-stabilization of byzantine protocols.
\newblock In Ted Herman and S\'{e}bastien Tixeuil, editors, {\em
  Self-Stabilizing Systems}, volume 3764 of {\em Lecture Notes in Computer
  Science}, pages 48--67. Springer, 2005.

\bibitem{D74j}
Edsger~W. Dijkstra.
\newblock Self-stabilizing systems in spite of distributed control.
\newblock {\em Commun. ACM}, 17(11):643--644, 1974.

\bibitem{DH07cb}
Danny Dolev and Ezra~N. Hoch.
\newblock On self-stabilizing synchronous actions despite byzantine attacks.
\newblock In Andrzej Pelc, editor, {\em DISC}, volume 4731 of {\em Lecture
  Notes in Computer Science}, pages 193--207. Springer, 2007.

\bibitem{D00b}
S.~Dolev.
\newblock {\em Self-stabilization}.
\newblock MIT Press, March 2000.

\bibitem{DW04j}
Shlomi Dolev and Jennifer~L. Welch.
\newblock Self-stabilizing clock synchronization in the presence of byzantine
  faults.
\newblock {\em J. ACM}, 51(5):780--799, 2004.

\bibitem{GS99c}
Mohamed~G. Gouda and Marco Schneider.
\newblock Stabilization of maximal metric trees.
\newblock In Anish Arora, editor, {\em WSS}, pages 10--17. IEEE Computer
  Society, 1999.

\bibitem{GS03j}
Mohamed~G. Gouda and Marco Schneider.
\newblock Maximizable routing metrics.
\newblock {\em IEEE/ACM Trans. Netw.}, 11(4):663--675, 2003.

\bibitem{HDD06c}
Ezra~N. Hoch, Danny Dolev, and Ariel Daliot.
\newblock Self-stabilizing byzantine digital clock synchronization.
\newblock In Ajoy~Kumar Datta and Maria Gradinariu, editors, {\em SSS}, volume
  4280 of {\em Lecture Notes in Computer Science}, pages 350--362. Springer,
  2006.

\bibitem{LSP82j}
Leslie Lamport, Robert~E. Shostak, and Marshall~C. Pease.
\newblock The byzantine generals problem.
\newblock {\em ACM Trans. Program. Lang. Syst.}, 4(3):382--401, 1982.

\bibitem{MT07j}
Toshimitsu Masuzawa and S\'{e}bastien Tixeuil.
\newblock Stabilizing link-coloration of arbitrary networks with unbounded
  byzantine faults.
\newblock {\em International Journal of Principles and Applications of
  Information Science and Technology (PAIST)}, 1(1):1--13, December 2007.

\bibitem{NA02c}
Mikhail Nesterenko and Anish Arora.
\newblock Tolerance to unbounded byzantine faults.
\newblock In {\em 21st Symposium on Reliable Distributed Systems (SRDS 2002)},
  page~22. IEEE Computer Society, 2002.

\bibitem{SOM05c}
Yusuke Sakurai, Fukuhito Ooshita, and Toshimitsu Masuzawa.
\newblock A self-stabilizing link-coloring protocol resilient to byzantine
  faults in tree networks.
\newblock In {\em Principles of Distributed Systems, 8th International
  Conference, OPODIS 2004}, volume 3544 of {\em Lecture Notes in Computer
  Science}, pages 283--298. Springer, 2005.

\bibitem{T09bc}
S\'{e}bastien Tixeuil.
\newblock {\em Algorithms and Theory of Computation Handbook, Second Edition},
  chapter Self-stabilizing Algorithms, pages 26.1--26.45.
\newblock Chapman \& Hall/CRC Applied Algorithms and Data Structures. CRC
  Press, Taylor \& Francis Group, November 2009.

\end{thebibliography}

\end{document}